\pdfoutput=1
\documentclass{amsart}
\usepackage[foot]{amsaddr}
\usepackage{amsmath}
\usepackage{amsfonts}
\usepackage{amssymb}
\usepackage[overload]{empheq}
\usepackage{lineno}
\usepackage{enumerate}
\usepackage{graphicx}
\usepackage{pgfplots}
\usepackage{mathptmx}      

\usepackage{color}
\usetikzlibrary{automata, positioning, arrows, calc,shapes}

\newcommand{\ggam}[2]{\Gamma_{#1}(#2)}

\DeclareSymbolFont{rsfscript}{OMS}{rsfs}{m}{n}
\DeclareSymbolFontAlphabet{\mathrsfs}{rsfscript}

\DeclareMathOperator{\dt}{.}

\DeclareMathOperator{\excl}{\mathrm{excl}}
\DeclareMathOperator{\dupl}{\mathrm{dupl}}
\DeclareMathOperator{\Cay}{\mathrm{Cay}}

\newcommand{\sa}{synchronizing automata}
\newcommand{\san}{synchronizing automaton}
\newcommand{\cra}{completely reachable automata}
\newcommand{\cran}{completely reachable automaton}

\newcommand{\scc}{strongly connected component}
\newcommand{\scn}{strongly connected}

\newcommand{\orb}{\mathrm{orb}} 

\newcommand{\mA}{\mathrsfs{A}}
\newcommand{\mB}{\mathrsfs{B}}

\newcommand{\mE}{\mathrsfs{E}}
\newcommand{\Z}{\mathbb{Z}}

\newcommand{\Zn}{\mathbb{Z}_n}

\newtheorem{theorem}{Theorem}
\newtheorem{proposition}[theorem]{Proposition}
\newtheorem{lemma}[theorem]{Lemma}

\theoremstyle{remark}
\newtheorem*{remark}{Remark}
\newtheorem{example}{Example}

\usepackage{url}
\pgfplotsset{compat=1.16}
\begin{document}

\title[Don's conjecture for binary completely reachable automata]{Don's conjecture\\ for binary completely reachable automata:\\
an approach and its limitations}

\author[D. Casas]{David Casas}
\address{{\normalfont Institute of Natural Sciences and Mathematics, Ural Federal University, 620000 Ekaterinburg, Russia}}
\email{dafecato4@gmail.com}
\email{m.v.volkov@urfu.ru}

\author[M. V. Volkov]{Mikhail V. Volkov}

\thanks{The authors were supported by the Ministry of Science and Higher Education of the Russian Federation, project FEUZ-2023-0022.}

\begin{abstract}
A deterministic finite automaton in which every non-empty set of states occurs as the image of the whole state set under the action of a suitable input word is called completely reachable. It was conjectured that in each completely reachable automaton with $n$ states, every set of $k>0$ states is the image of a word of length at most $n(n-k)$. We confirm the conjecture for completely reachable automata with two input letters satisfying certain restrictions on the action of the letters.
\end{abstract}


\maketitle

\section{Background and overview}
\label{sec:intro}

A \emph{complete deterministic finite automaton} (DFA) is a triple $\mathrsfs{A}=\langle Q,\Sigma,\delta\rangle$ where $Q$ and $\Sigma$ are finite non-empty sets and $\delta\colon Q\times\Sigma\to Q$ is an everywhere defined map called the \emph{transition function} of $\mA$.

The elements of $Q$ and $\Sigma$ are called \emph{states} and, resp., \emph{letters}. \emph{Words over $\Sigma$} are finite sequences of letters (including the empty sequence denoted by $\varepsilon$). The set of all words over $\Sigma$ is denoted by $\Sigma^*$. We define the \emph{length} of a word $w=a_1\cdots a_n$ with $a_i\in\Sigma$, $i=1,\dots,n$, as the number $n$ and assume that the length of $\varepsilon$ is 0.

The transition function $\delta$ extends to a function $Q\times\Sigma^*\to Q$ (still denoted by $\delta$) via the following recursion: for every $q\in Q$, we set $\delta(q,\varepsilon)=q$  and $\delta(q,wa)=\delta(\delta(q,w),a)$ for all $w\in\Sigma^*$ and $a\in\Sigma$. Thus, every word $w\in\Sigma^*$ induces the transformation $q\mapsto\delta(q,w)$ on the set $Q$. The set $T(\mA)$ of all transformations induced this way is closed under the composition of transformations and contains the identity transformation. Thus, $T(\mA)$ is a monoid called the \emph{transition monoid} of $\mA$.

Let  $\mathrsfs{A}=\langle Q,\Sigma,\delta\rangle$ be a DFA. The transition function $\delta$ can be further extended to non-empty subsets of the set $Q$. Namely, for every non-empty subset $P\subseteq Q$ and every word $w\in\Sigma^*$, we let $\delta(P,w)=\{\delta(q,w)\mid q\in P\}$.

We lighten the above notation by suppressing the sign of the transition function: we specify a DFA as a pair $\langle Q,\Sigma\rangle$ and write $q\dt w$ for $\delta(q,w)$ and $P\dt w$ for $\delta(P,w)$.

We say that a non-empty subset $P\subseteq Q$ is \emph{reachable} in $\mathrsfs{A}=\langle Q,\Sigma\rangle$ if $P=Q\dt w$ for some word $w\in\Sigma^*$. A DFA is called \emph{completely reachable} if every non-empty subset of its state set is reachable. Observe that complete reachability is actually a property of the transition monoid of $\mA$; hence, if a DFA $\mathrsfs{B}$ with the same state set as $\mA$ is such that $T(\mB)=T(\mA)$, then $\mB$ is completely reachable if and only if so is $\mA$.

Henk Don~\cite[Conjecture 18]{Don16} conjectured that if in a DFA with $n$ states, some subset $S$ with $k>0$ states is reachable, then $S$ is reachable by a word of length $\le n(n-k)$. In this form, the conjecture was far too bold as its validity would imply the collapse of polynomial hierarchy; see \cite[Section 2.3]{GoJu19} for a discussion. Moreover, Fran{\c{c}}ois Gonze and Rapha{\"e}l Jungers constructed a series of $n$-state automata with a~distinguished subset $S$ of size $\lfloor\frac{n}2\rfloor$ such that for each $n>6$, the length of the shortest word that reaches $S$ is greater than $\frac{2^n}{n}$; see \cite[Proposition 7]{GoJu19}. However, the restriction of the conjecture to \cra\ makes sense, and to the best of our knowledge, the question of whether this restriction holds remains open\footnote{See the note at the end of the paper added when the paper was revised.}. Robert Ferens and Marek Szyku\l{}a \cite{FeSz22} proved that in a completely reachable automaton with $n$ states, each subset with $k>0$ states is reachable by a word of length $\le 2n(n-k)$. Here we aim to prove Don's conjecture for \cra\ with two letters under certain restrictions on the letters' action.

DFAs with two letters are called \emph{binary}. Except for the 2-state flip-flop, a binary \cran\ is always \emph{circular}, i.e., one of the letters acts as a cyclic permutation of all states; see \cite[Lemma 1]{CaVo22}. Throughout, the letter acting as a cyclic permutation is denoted by $b$ while the other letter is denoted by $a$. Without any loss, we assume that all circular DFAs with $n$ states have the set $\Zn$ of all residues modulo $n$ as their state sets, and the letter $b$ acts on $\Zn$ by adding 1 modulo $n$, i.e., for each $q\in\Zn$, we have $q\dt b=q\oplus1$ where $\oplus$ stands for addition modulo $n$.

Given a DFA $\mathrsfs{A}=\langle Q,\Sigma\rangle$ and a word $w\in\Sigma^*$, we call the set $\excl(w):=Q{\setminus}Q\dt w$ the \emph{excluded set} of $w$ and the set $\dupl(w):=\{p \mid p=q_1\dt w=q_2\dt w\text{ for some } q_1\ne q_2\}$ the \emph{duplicate set} of $w$. For the complete reachability of a binary circular DFA $\langle\Zn,\{a,b\}\rangle$, it is necessary that the sets $\excl(a)$ and $\dupl(a)$ both are singletons. We may assume that $\excl(a)=\{0\}$ as it does not matter from where the cyclic count of the states starts. Further, the DFAs $\mA=\langle\Zn,\{a,b\}\rangle$ and $\mA_k=\langle \Zn,\{a_k,b\}\rangle$, where the letter $a_k$ acts in $\mA_k$ as the word $b^ka$ does in $\mA$, are easily seen to have the same transition monoid. Hence, $\mA$ and $\mA_k$ are or are not completely reachable at the same time. Choosing $k=q_1$ (or $k=q_2$) where $q_1\ne q_2$ are such that $q_1\dt a=q_2\dt a$, we get $\{0\dt a_k\}=\dupl(a)=\dupl(a_k)$ while $\excl(a_k)=\excl(a)=\{0\}$. Thus, passing to the DFA $\langle \Zn,\{a_k,b\}\rangle$ and writing $a$ instead of $a_k$, we get $\excl(a)=\{0\}$ and $\dupl(a)=\{0\dt a\}$. We call a circular binary DFA $\langle \Zn,\{a,b\}\rangle$ \emph{standardized} if the letter~$a$ satisfies these two assumptions.

Denote the state $0\dt a$ by $d$ and let $r$ stand for the state such that $r\ne0$ and $r\dt a=d$. In a {standardized} DFA, the letter $a$ acts as a permutation on the set $\{1,\dots,n-1\}$; see the proof of Proposition~3 in \cite{CaVo22}. Therefore, acting by a suitable power of the letter $a$ at the state $d$, one gets the state $r$, that is, $r=d\dt a^{\ell-1}$ for some positive integer $\ell$. Let $\ell$ be the least positive integer with this property, and for each $s=0,1,\dots,\ell-1$, let $d_s:=d\dt a^s$ so that $d_0=d$, $d_{\ell-1}=r$, and all states $d_0,d_1,\dots,d_{\ell-1}$ are distinct. We denote the set $\{d_0,d_1,\dots,d_{\ell-1}\}$ by $\orb(d)$ and call it the \emph{orbit} of the DFA $\langle\Zn,\{a,b\}\rangle$. The subgroup of $(\Zn,\oplus)$ generated by $\orb(d)$ is called the \emph{orbit subgroup} of the DFA.
The binary DFA $\mE_6$ presented in Figure~\ref{fig:example concepts} illustrates the above concepts. In this particular case, $d_0 = 2$, $r = 5$, and $\orb(2) = \{2,5\}$ whence $\ell = 2$.

\begin{figure}[h]
\begin{center}
\begin{tikzpicture}
	\foreach \ev in {0,1,2,3,4,5}
	{
		\node[fill=white, circle, draw=blue, scale=1] at ($({60*\ev}:2cm)$) (\ev) {$\ev$};
	}
	
	\draw
	(0) edge[-latex, dashed] (1)
	(1) edge[-latex, dashed] (2)
	(2) edge[-latex, dashed] (3)
	(3) edge[-latex, dashed] (4)
	(4) edge[-latex, dashed] (5)
	(5) edge[-latex, dashed] (0)
	
	(0) edge[-latex] (2)
	(2) edge[-latex, bend right] (5)
	(5) edge[-latex] (2)
	;
\end{tikzpicture}
\end{center}
\caption{The DFA $\mE_{6}$. Solid and dashed edges show the action of $a$ and, resp., $b$; if $a$ fixes a state, the corresponding loop is omitted to improve readability.}
\label{fig:example concepts}
\end{figure}
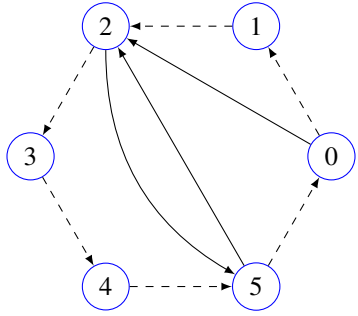

The main result of this note is the following:
\begin{theorem}
\label{thm:good orbit}
Every standardized DFA $\langle\Zn,\{a,b\}\rangle$ whose orbit subgroup coincides with the group $(\Zn,\oplus)$ fulfills Don's conjecture.
\end{theorem}

Theorem~\ref{thm:good orbit} generalizes Don's result in \cite[Proposition 15]{Don16} that establishes the conjecture provided the numbers $d$ and $n$ are coprime, that is, already the subgroup generated by $d$ alone coincides with $(\Zn,\oplus)$.

The note is organized as follows. In Section~\ref{sec:extension}, we present a sufficient condition for a general DFA $\mathrsfs{A}=\langle Q,\Sigma\rangle$ to satisfy Don's conjecture; this condition is stated in terms of a property of subsets of $Q$ called expandability. The core of the note is Section~\ref{sec:orbit digraph} where we show that in every standardized DFA $\langle\Zn,\{a,b\}\rangle$, all non-empty subsets are $n$-expandable except for, perhaps, unions of cosets of the orbit subgroup of the DFA. This readily implies Theorem~\ref{thm:good orbit}. In Section~\ref{sec:limitations}, we analyze our approach and present examples demonstrating its limitations. In Section~\ref{sec:conclusion}, we summarize our results and discuss their similarity to certain facts concerning \sa.

We assume the reader's acquaintance with basic concepts concerning directed graphs (digraphs), such as strong connectivity, (directed) cycle and spanning subgraph, and the notion of a coset of a subgroup in a group.

\section{Expandable subsets}
\label{sec:extension}
The two results collected in this section are basically known as their versions have been scattered over the literature; see, e.g., \cite{Don16,FeSz22,BCV22}. We believe it is more convenient for the reader to see direct arguments rather than follow references to various sources where similar ideas might have appeared under different terminology and notation. Therefore, we have included complete proofs without claiming any originality.

Let $\mathrsfs{A}=\langle Q,\Sigma\rangle$ be a DFA. We say that a word $w\in\Sigma^*$ \emph{expands} a proper non-empty subset $S\subset Q$ if there exists a subset $P\subseteq Q$ such that $|P|>|S|$ and $P\dt w=S$. The following easy observation connects this notion with the concepts of excluded and duplicated sets introduced in Section~\ref{sec:intro}.

\begin{lemma}\label{lem:extension}
Let $\mathrsfs{A}=\langle Q,\Sigma\rangle$ be a DFA. A word $w\in\Sigma^*$ expands a proper non-empty subset $S\subset Q$ if and only if $\excl(w)\cap S=\varnothing$ while $\dupl(w)\cap S\ne\varnothing$.
\end{lemma}

\begin{proof}
For the `only if' part, let $P$ be a subset of $Q$ with $|P|>|S|$ and $P\dt w=S$. Since $S=P\dt w\subseteq Q\dt w$, we get $(Q{\setminus}Q\dt w)\cap S=\varnothing$, that is, $\excl(w)\cap S=\varnothing$. Since $|P|>|P\dt w|$, there exist some $p,p'\in P$ such that $p\ne p'$ but $p\dt w$ coincides with $p'\dt w$. Then $p\dt w\in\dupl(w)\cap S$ whence $\dupl(w)\cap S\ne\varnothing$.

Conversely, for the `if' part, let $w\in\Sigma^*$ be a word with $\excl(w)\cap S=\varnothing$ and $\dupl(w)\cap S\ne\varnothing$. Since $\excl(w)=Q{\setminus}Q\dt w$ is disjoint from $S$, we have $S\subseteq Q\dt w$. Hence for every state $s\in S$, its preimage $sw^{-1}:=\{q\in Q\mid q\dt w=s\}$ is non-empty. Let $P:=\bigcup\limits_{s\in S} sw^{-1}$. Then $P\dt w=S$ and $|P|>|S|$ since the subsets $sw^{-1}$ are disjoint and for each $p\in\dupl(w)\cap S$, the set $pw^{-1}$ is non-singleton.
\end{proof}

Given a DFA $\mathrsfs{A}=\langle Q,\Sigma\rangle$ with $|Q|=n$, a proper non-empty subset of $Q$ is said to be $n$-\emph{expandable} if it can be expanded by a word of length at most $n$.

\begin{lemma}\label{lem:n-extension}
If in a DFA $\langle Q, \Sigma \rangle$ with $n$ states, every proper non-empty subset of $Q$ is $n$-expandable, then every subset with $k>0$ states is reachable by a word of length $\le n(n-k)$.
\end{lemma}

\begin{proof}
We prove that for any $k$ with $0<k\le n$, every subset $S\subseteq Q$ with $k$ states is reachable by a word of length $\le n(n-k)$ by induction on $n-k$.

If $n-k=0$, then $S=Q$ and the claim holds since $Q$ is reachable by the empty word whose length is 0.

Now let $n-k>0$ so that $S$ is a proper subset of $Q$. Then $S$ is $n$-expandable so that there exist a word $w\in\Sigma^*$ of length at most $n$ and a subset $P\subseteq Q$ such that $|P|>|S|$ and $P\dt w=S$. Since  $|P|>|S|=k$, we have $n-|P|<n-k$, and the induction assumption applies to the subset $P$. Hence, $P=Q\dt v$ for some word $v\in\Sigma^*$ of length $\le n(n-k-1)$. Then $S=P\dt w=(Q\dt v)\dt w=Q\dt vw$ and the length of the word $vw$ does not exceed $n(n-k-1)+n=n(n-k)$ as required.
\end{proof}

Lemmas \ref{lem:extension} and  \ref{lem:n-extension} imply that a DFA $\mathrsfs{A}=\langle Q,\Sigma\rangle$ satisfies Don's conjecture whenever for each proper non-empty subset $S\subset Q$, one can find a word $w$ of length at most $|Q|$ with $\excl(w)\cap S=\varnothing$ and $\dupl(w)\cap S\ne\varnothing$.

One caveat seems to be in order: our notion of expandability should not be confused with {extensibility}, a similar but different concept widely used in the theory of synchronizing automata. We discuss this in some detail in the concluding section.

\section{The restricted orbit digraph}
\label{sec:orbit digraph}

Recall that if $X$ is a subset in a group $G$, the (\emph{right}) \emph{Cayley digraph of $G$ with respect to $X$}, denoted $\Cay(G,X)$, has $G$ as its vertex set and $\{(g,gx)\mid g\in G,\ x\in X\}$ as its edge set. The following property of Cayley digraphs of finite groups is folklore\footnote{In fact, our definition is the semigroup version of the notion of a Cayley digraph, but this makes no difference since in a finite group, every subsemigroup is a subgroup.}.
\begin{lemma}
\label{lem:cayley}
Let $G$ be a finite group, $X$ a subset of $G$, and $H$ the subgroup of $G$ generated by $X$. The \scc{}s\ of the Cayley digraph\/ $\Cay(G,X)$ have the left cosets $gH$, $g\in G$, as their vertex sets, and each \scc\ is isomorphic to $\Cay(H,X)$.
\end{lemma}

The \emph{orbit digraph} $\mathcal{O}(\mA)$ of a standardized DFA $\mA=\langle \Zn,\{a,b\}\rangle$ is the Cayley digraph $\Cay(\Zn,\orb(d))$. Denote the orbit subgroup of $\mA$ by $H_0$. Thus, each edge of $\mathcal{O}(\mA)$ is of the form $q\to q\oplus d_s$, where $q\in\Zn$ and $d_s\in \orb(d)$, and the \scc{}s of $\mathcal{O}(\mA)$ have the cosets $ q\oplus H_0$,  $q\in\Zn$, as their vertex sets.

The edge $q\to q\oplus d_s$ is called \emph{long} if $q+s\ge n$ and \emph{short} otherwise. The \emph{restricted orbit digraph} $\mathcal{R}(\mA)$ is the spanning subgraph of the orbit digraph obtained by removing all long edges from the latter digraph. We aim to show that what was connected in the orbit digraph remains so in the restricted orbit digraph. For this, we need an elementary lemma.
\begin{lemma}
\label{lem:arithmetic}
Let $n$ be an arbitrary positive integer and $d_0,d_1,\dots,d_{s-1}$ be $s \geq 1$ distinct positive integers all less than $n$. Then the greatest common divisor of $d_0,d_1,\dots,d_{s-1},n$ does not exceed $\frac{n}{s+1}$.
\end{lemma}

\begin{proof}
Let $c_0,c_1,\dots,c_{s-1}$ be the numbers $d_0,d_1,\dots,d_{s-1}$ arranged in the ascending order. Denoting the greatest common divisor of $d_0,d_1,\dots,d_{s-1},n$ by~$g$, we get $c_0\ge g$, $c_1\ge 2g$, \dots, $c_{s-1}\ge sg$, and finally, $n\ge(s+1)g$ since each of the numbers $c_0,c_1,\dots,c_{s-1},n$ is a multiple of $g$ and all these numbers are distinct. Hence, $g\le\frac{n}{s+1}$.
\end{proof}
\begin{proposition}
\label{prop:restricted scc}
Let  $\mA=\langle \Zn,\{a,b\}\rangle$ be a standardized DFA, and $H_0$ the subgroup of $(\Zn,\oplus)$ generated by $\orb(d)$. The \scc{}s of the restricted orbit digraph $\mathcal{R}(\mA)$ have the cosets of the subgroup $H_0$ as their vertex sets.
\end{proposition}

\begin{proof}
For a state $p\in\Zn$, denote by $\langle p\rangle$ the subgroup of $(\Zn,\oplus)$ generated by $p$. Let $\orb(d)=\{d_0,d_1,\dots,d_{\ell-1}\}$, and for each $s=0,1,\dots,\ell-1$, let $g_s$ stand for the greatest common divisor of the numbers $d_0,d_1,\dots,d_s,n$. We then have that $g_0$ is a multiple of $g_1$, which is a multiple of $g_2$, and so on.

Inducting on $s$, we will establish the following:
\begin{list}{\textbf{Claim:}}{}
\item \emph{For each $s=0,1,\dots,\ell-1$, the restricted orbit digraph $\mathcal{R}(\mA)$ has a spanning subgraph $\Gamma^{(s)}$ whose \scc{}s have the cosets of the subgroup $\langle g_s\rangle$ as their vertex sets.}
\end{list}

\noindent\emph{Proof of the claim.} For $s=0$, consider the following $n$ edges in $\mathcal{O}(\mathrsfs{A})$:
\[
0\to d_0,\ 1\to d_0\oplus1,\ \dots,\ n-1\to d_0\oplus(n-1).
\]
They all are short and easily seen to form $g_0$ directed cycles whose vertex sets are the cosets of the subgroup $\langle g_0\rangle$. Thus, the spanning subgraph with these $n$ edges can be taken as $\Gamma^{(0)}$.

Now let $s>0$. Construct a new spanning subgraph $\Gamma$ of  the graph $\mathcal{O}(\mathrsfs{A})$ by adding to the digraph $\Gamma^{(s-1)}$ the following $g_{s-1}$ edges:
\[
0\to d_s,\ \ 1\to d_s\oplus1,\ \dots,\ g_{s-1}-1\to d_s\oplus (g_{s-1}-1).
\]
First, we verify that each of these edges is short. For this, it suffices to show that $s+(g_{s-1}-1)<n$, that is, $g_{s-1}+s\le n$. Here we make use of Lemma~\ref{lem:arithmetic} to deduce the required inequality $g_{s-1}+s\le n$. Indeed,
\begin{align*}
n-g_{s-1}-s&\ge n-\frac{n}{s+1}-s&&\text{by Lemma~\ref{lem:arithmetic}}\\
           &=n\left(1-\frac1{s+1}\right)-s&&\\
           &\ge(s+1)\left(1-\frac1{s+1}\right)-s&&\text{since $n>\ell\ge s$}\\
           &=(s+1)-1-s=0.&&
\end{align*}

Back to the construction of the spanning subgraph $\Gamma$, we have to analyze its \scc{}s.

It readily follows from the definition of the greatest common divisor that $g_s$ is the greatest common divisor of $g_{s-1}$ and $d_s$. Hence $g_{s-1}=mg_s$ and $d_s=kg_s$ for some coprime $m$ and $k$. The subgroup $\langle g_s\rangle$ is equal to the union of the $m$ cosets of the subgroup $\langle g_{s-1}\rangle$ that are contained in $\langle g_s\rangle$; these $m$ cosets are
\begin{equation}\label{eq:cosets}
\langle g_{s-1}\rangle,\ \ g_s\oplus\langle g_{s-1}\rangle,\ \dots,\ (m-1)g_s\oplus\langle g_{s-1}\rangle.
\end{equation}

We have $d_s\in\overline{k}g_s\oplus\langle g_{s-1}\rangle$, where $\overline{k}$ is the residue of $k$ modulo $m$. Therefore, in the subgraph $\Gamma$, the newly added edge $0\to d_s$ connects the \scc{}s $\langle g_{s-1}\rangle$ and $\overline{k}g_s\oplus\langle g_{s-1}\rangle$ of the subgraph $\Gamma^{(s-1)}$. In the same way, the edge $\overline{k}g_s\to d_s\oplus\overline{k}g_s$ connects the \scc{}s $\overline{k}g_s\oplus\langle g_{s-1}\rangle$ and $\overline{2k}\oplus\langle g_{s-1}\rangle$, where $\overline{2k}$ is the residue of $2k$ modulo $m$, etc. Since $m$ and $k$ are coprime, the $m$ edges
\[
0\to d_s,\ \overline{k}g_s\to d_s\oplus\overline{k}g_s,\ \overline{2k}g_s\to d_s\oplus\overline{2k}g_s,\ \dots,\ \overline{(m-1)k}g_s\to d_s\oplus \overline{(m-1)k}g_s
\]
cyclically connect all $m$ cosets in~\eqref{eq:cosets}. By the induction assumption, each of these cosets is the vertex set of a \scc\ of the digraph $\Gamma^{(s-1)}$. Hence, all states in the subgroup $\langle g_s\rangle$ are mutually reachable in the digraph $\Gamma$.

In the same way, for each $i=1,\dots,g_s-1$, the $m$ edges
\begin{multline*}
i\to d_s\oplus i,\ \overline{k}g_s\oplus i\to d_s\oplus \overline{k}g_s\oplus i,\ \overline{2k}g_s\oplus i\to d_s\oplus \overline{2k}g_s\oplus i,\ \dots,\\
\overline{(m-1)k}g_s\oplus i\to d_s\oplus \overline{(m-1)k}g_s\oplus i
 \end{multline*}
cyclically connect the $m$ cosets
\[
i\oplus \langle g_{s-1}\rangle,\ i\oplus g_s\oplus \langle g_{s-1}\rangle,\ \dots,\ i\oplus (m-1)g_s\oplus \langle g_{s-1}\rangle.
\]
As above, using the induction assumption, we conclude that all states in the coset $i+\langle g_s\rangle$ also are mutually reachable in the digraph $\Gamma$. Since no more edges were added when constructing the graph $\Gamma$, the $g_s$ cosets $i+\langle g_s\rangle$ with $i=0,1,\dots,g_s-1$ form the vertex sets of the \scc{}s of $\Gamma$.

We have verified that the spanning subgraph $\Gamma$ fulfils all requirements we need, and thus, can be taken as $\Gamma^{(s)}$. This completes the proof of the inductive step, and hence, the proof of the claim.\qed

\begin{remark}
It may happen that $g_{s-1}$ divides $d_s$, in which case $g_{s-1}=g_s$. In this situation the above construction of the spanning subgraph $\Gamma^{(s)}$ still works fine (with $m=1$), because each newly added edge connects vertices within a  \scc{} of $\Gamma^{(s-1)}$. Thus, while having more edges, the graph $\Gamma^{(s)}$ has the same \scc{}s as $\Gamma^{(s-1)}$.
\end{remark}

The proof of Proposition~\ref{prop:restricted scc} is now immediate since the subgroup $\langle g_{\ell-1}\rangle$ coincides with the subgroup $H_0$.
\end{proof}

We will use the following notation. If $m$ is an arbitrary element of the cyclic group $\Z_n$, the subgroup of $\Z_n$ generated by $m$, i.e., $\{0, m, 2m, \dots \}$, is denoted by $m\Z_n$. For an element $q\in\Z_n$, the coset of $m\Z_n$ that contains $q$ is denoted by $q \oplus m\Z_n$.

For an illustration, we trace the inductive construction in the proof of Proposition~\ref{prop:restricted scc} on the 48-state automaton $\mE_{48}=\langle\Z_{48},\{a,b\}\rangle$ shown in Figure~\ref{fig:e48} below. In $\mE_{48}$, we have $d=24$, and the orbit of $\mE_{48}$ consists of $d=d_0$ and $r=d_1=18$ so that $\ell=2$.

\begin{figure}[htb]
\begin{center}
  \begin{tikzpicture}
  [scale=0.95]
         \pgfmathsetmacro{\n}{24}
        \foreach \t [evaluate=\t as \teval using int(2*\t)] in {1,...,23} {
        \edef\temp{\noexpand
        \node[fill=white, circle, draw=blue, scale=1] (\teval) at ( {4*cos((360*\t)/\n)}, {4*sin((360*\t)/\n)} )  {\teval};
         }\temp}
         \foreach \t [evaluate=\t as \teval using int(2*\t-1)] in {1,...,\n} {
         \edef\temp{\noexpand
         \node[fill=white, circle, draw=blue, scale=1] (\teval) at ( {5.5*cos((360*\t)/\n-7.5)}, {5.5*sin((360*\t)/\n-7.5)} ) {\teval};
         }\temp}
         \node[fill=white, circle, draw=blue, scale=1] (0) at ( {4}, {0} ) {$0$};
		 \foreach \t [evaluate=\t as \teval using int(\t + 1)]  in {0,1,...,46}{
		 	\draw
		 	(\t) edge[-latex, densely dotted] (\teval);
         }
         \draw
         	(47) edge[-latex, densely dotted] (0)
         	(0) edge[-latex] (24)
         	(18) edge[-latex, bend left = 40] (24)
         	(24) edge[-latex, bend right = 15] (18)
         	(14) edge[-latex, bend right=20] (13)
         	(13) edge[-latex, bend right=20] (14)
            (30) edge[-latex, bend right=40] (32)
         	(32) edge[-latex, bend right=40] (30);
         \end{tikzpicture}
\end{center}
\caption{The DFA $\mE_{48}=\langle\Z_{48},\{a,b\}\rangle$. Solid and dotted edges show the action of $a$ and, resp., $b$; if $a$ fixes a state, the corresponding loop is omitted to improve readability.}\label{fig:e48}
\end{figure}

We have $g_0=d_0=24$. Therefore, we begin the construction with the spanning subgraph $\Gamma^{(0)}$ that consists of the 24 directed cycles

\smallskip

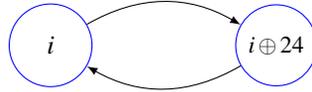
\begin{center}
\begin{tikzpicture}
[scale=0.9]
	\node[fill=white, circle, draw=blue, minimum size=1.1cm] (0) {$i$};
	\node[fill=white, circle, draw=blue, right of = 0, xshift= 2cm] (1) {\small$i\oplus 24$};
	\draw
		(0) edge[-latex, bend left] (1)
		(1) edge[-latex, bend left] (0);
\end{tikzpicture}
\end{center}

\smallskip

\noindent with $i=0,1,\dots,23$, having the 24 cosets of the 2-element subgroup $24\Z_{48}=\{0,24\}$ as the vertex sets. All edges in $\Gamma^{(0)}$ are short.

To get the next spanning subgraph $\Gamma^{(1)}$, we add to $\Gamma^{(0)}$ the following 24 edges:
\[
0\to 18,\ 1\to 19,\ 2\to 20,\ \dots,\ 23\to 41.
\]
which are all short. The greatest common divisor $g_1$ of the numbers $d_0=24$, $d_1=18$ and $48$ is 6. We have $m=\dfrac{g_0}{g_1}=4$ and $k=\dfrac{d_1}{g_1}=3$. The 8-element subgroup $\langle g_s\rangle=6\Z_{48}=\{0,6,12,18,24,30,36,42\}$ is the union of the following four cosets of the group $(24\Z_{48},\oplus)$:
\[
24\Z_{48},\ \ 6\oplus24\Z_{48},\ \ 12\oplus24\Z_{48},\ \ 18\oplus24\Z_{48}.
\]
The newly added edges $0\to 18$, $18\to 36$, $6\to 24$, $12\to 30$ cyclically connect these four cosets, producing a \scc{} of $\Gamma^{(1)}$ as shown in Figure~\ref{fig:scc8}. The three other  \scc{}s of the digraph $\Gamma^{(1)}$ are constructed in the same way.
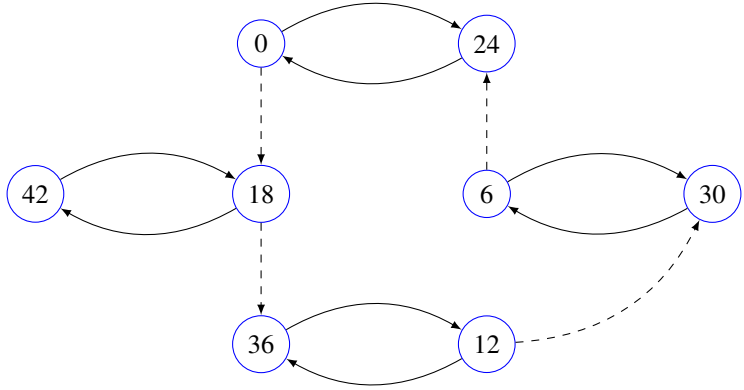
\begin{figure}[htb]
\begin{center}
  \begin{tikzpicture}
  [scale=0.9]
\node[fill=white, circle, draw=blue] (0) {0};
\node[fill=white, circle, draw=blue, right of = 0, xshift= 2cm] (24) {24};
\node[fill=white, circle, draw=blue, below of = 24, yshift=-1cm] (6) {6};
\node[fill=white, circle, draw=blue, below of = 0, yshift=-1cm] (18) {18};
\node[fill=white, circle, draw=blue, right of = 6, xshift= 2cm] (30) {30};
\node[fill=white, circle, draw=blue, below of = 6, yshift=-1cm] (12) {12};
\node[fill=white, circle, draw=blue, left of = 12, xshift= -2cm] (36) {36};
\node[fill=white, circle, draw=blue, left of = 18, xshift= -2cm] (42) {42};
	\draw
		(0) edge[-latex, bend left] (24)
		(24) edge[-latex, bend left] (0)
        (30) edge[-latex, bend left] (6)
		(6) edge[-latex, bend left] (30)
        (36) edge[-latex, bend left] (12)
		(12) edge[-latex, bend left] (36)
        (42) edge[-latex, bend left] (18)
		(18) edge[-latex, bend left] (42)
        (6) edge[-latex, dashed] (24)
        (0) edge[-latex, dashed] (18)
        (18) edge[-latex, dashed] (36)
        (12) edge[-latex, bend right, dashed] (30);
  \end{tikzpicture}
\end{center}
\caption{One of the \scc{}s of the digraph $\Gamma^{(1)}$ constructed for the DFA $\mE_{48}$ from Figure~\ref{fig:e48}. The solid edges are inherited from $\Gamma^{(0)}$; the dashed edges are newly added.}\label{fig:scc8}
\end{figure}

As an application of Proposition~\ref{prop:restricted scc}, we infer that certain proper non-empty subsets in standardized DFAs are $n$-expandable.
\begin{proposition}\label{prop:noncosets}
Let $\mA=\langle\Zn,\{a,b\}\rangle$ be a standardized DFA and $H_0$ its orbit subgroup. Every non-empty subset of $\Zn$ which is not a union of $H_0$-cosets is $n$-expandable.
\end{proposition}

\begin{proof}
If a non-empty subset $S$ of $\Zn$ is not a union of $H_0$-cosets, then there exists a~coset $C$ which is neither contained in $S$ nor disjoint from $S$. Then there exist a~state $p\in C{\setminus}S$ and a state $p'\in C\cap S$. Proposition~\ref{prop:restricted scc} implies that any two states in $C$ are connected in the restricted orbit digraph $\mathcal{R}(\mA)$; in particular, there is a sequence $p=p_0,p_1,\dots,p_t=p'$ of states in $C$ such that $p_{i-1}\to p_i$ is an edge in $\mathcal{R}(\mA)$ for each $i=1,\dots,t$. If $j$ is the maximal index such that $p_j\in C{\setminus}S$, then $j<t$ and $p_{j+1}\in C\cap S$. Renaming $p_j$ and $p_{j+1}$ to $q$ and $q'$, respectively, we conclude that the edge $q\to q'$ of $\mathcal{R}(\mA)$ is such that $q\notin S$ and $q'\in S$.

Let $\orb(d)=\{d_0,d_1,\dots,d_{\ell-1}\}$. By the construction of the restricted orbit digraph, $q\to q'$ being its edge means that $q\to q'$ is a short edge in the orbit digraph $\mathcal{O}(\mA)$. Unfolding the definitions of  $\mathcal{O}(\mA)$ and of being short, we see that $q'=q\oplus d_s$ for some $s\in\{0,1,\dots,\ell-1\}$ and $q+s<n$. Now consider the word $a^{s+1}b^q$ of length $q+s+1\le n$. Since
\[
\left.\begin{array}{l}
0\dt a^{s+1}b^q\\
r\dt a^{s+1}b^q
\end{array}
\right\}=d\dt a^{s}b^q=d_s\dt b^q=d_s\oplus q=q',
\]
the duplicate set of $a^{s+1}b^q$ contains $q'$. On the other hand, the only excluded state of $a^{s+1}$ is 0 whence $\excl(a^{s+1}b^q)=\{q\}$. Thus, we have $\excl(a^{s+1}b^q)\cap S=\varnothing$ while $\dupl(a^{s+1}b^q)\cap S\ne\varnothing$. Lemma~\ref{lem:extension} then implies that the word $a^{s+1}b^q$ expands $S$. Since the length of this word does not exceed $n$, the subset $S$ is $n$-expandable.
\end{proof}

Now we can easily deduce Theorem~\ref{thm:good orbit}.

\smallskip

\noindent\emph{Proof of Theorem~\ref{thm:good orbit}.}  Let $\mA=\langle\Zn,\{a,b\}\rangle$ be a standardized DFA whose orbit subgroup $H_0$ coincides with $(\Zn,\oplus)$. Then no non-empty proper subset of $\Zn$ can be a union of $H_0$-cosets, whence each non-empty proper subset of $\Zn$  is $n$-expandable by Proposition~\ref{prop:noncosets}. Now Lemma~\ref{lem:n-extension} implies that every subset with $k>0$ states is reachable by a word of length $\le n(n-k)$. Thus, the automaton $\mA$ fulfills Don's conjecture.\qed

\section{Limitations of our approach and its possible generalizations}
\label{sec:limitations}

Analyzing the above proof of Theorem~\ref{thm:good orbit}, we see that a stronger statement has actually been proved: in every standardized DFA $\mA=\langle\Zn,\{a,b\}\rangle$ whose orbit subgroup $H_0$ coincides with $(\Zn,\oplus)$, each subset with $k>0$ states is reachable by a product of $n-k$ words of the form $a^{i+1}b^j$ where $0\le i,j\le n-1$ and $i+j\le n$. Each word of the form $a^{i+1}b^j$ with $0\le i,j\le n-1$ has a unique excluded state (namely, $\excl(a^{i+1}b^j)=\{j\}$). For any word $w$, its \emph{defect} is defined as the cardinality of the set $\excl(w)$. Thus, under the premise of Theorem~\ref{thm:good orbit}, each subset with $k>0$ states is reachable by a product of $n-k$ words of defect 1 and length $\le n$.

For the convenience of the subsequent discussion, call a (not necessarily binary) DFA $\mathrsfs{A}=\langle Q,\Sigma\rangle$ \emph{perfectly reachable} if each subset with $k>0$ states is reachable in $\mA$ by a product of $|Q|-k$ words of defect 1. It turns out that this property can be characterized in terms of a certain digraph $\Gamma_1(\mathrsfs{A})$ associated with $\mA$. The digraph $\Gamma_1(\mathrsfs{A})$ has $Q$ as its vertex set and the edges of $\Gamma_1(\mathrsfs{A})$ are of the form $p\to q$ such that for some word $w\in\Sigma^*$ of defect 1, $\excl(w)=\{p\}$ while $\dupl(w)=\{q\}$; if so happens, we say that the edge $p\to q$ is \emph{forced} by the word $w$. The definition and notation come from~\cite{BondarVolkov16}, but much earlier, though in a less explicit form, the same digraph was used by Igor Rystsov~\cite{rystsov2000estimation} for some special species of DFAs. Therefore, we call $\Gamma_1(\mathrsfs{A})$ the \emph{Rystsov digraph} of $\mA$.

The following characterization of perfectly reachable automata arises as a combination of two results in the literature.

\begin{proposition}\label{prop:perfect}
A DFA is perfectly reachable if and only if its Rystsov digraph is \scn.
\end{proposition}

\begin{proof}
The proof of~\cite[Theorem~1]{BondarVolkov16} shows that a DFA $\mA$ is perfectly reachable whenever the digraph $\Gamma_1(\mathrsfs{A})$ is \scn. The converse follows from \cite[Theorem 20]{GoJu19}.
\end{proof}

Turning back to binary DFAs, observe that for any standardized DFA $\mA=\langle\Zn,\{a,b\}\rangle$, its orbit digraph is a spanning subgraph of the Rystsov digraph $\Gamma_1(\mathrsfs{A})$. Indeed, the edges of the orbit digraph are of the form $q\to q\oplus d_s$, where $q\in\Zn$ and $d_s\in \orb(d)$. The word $w_{s,q}:=a^{s+1}b^q$ has defect 1 and $\excl(w_{s,q})=q$ while $\dupl(w_{s,q})=q\oplus d_s$. Hence the edge $q\to q\oplus d_s$ occurs in the Rystsov digraph as the edge forced by $w_{s,q}$. By Lemma~\ref{lem:cayley}, DFAs satisfying the premise of Theorem~\ref{thm:good orbit} are precisely standardized DFAs with \scn\ orbit digraphs. Hence, such DFAs are perfectly reachable by Proposition~\ref{prop:perfect}.

Attempting to extend our approach to arbitrary, perfectly reachable standardized DFAs, one may define restricted versions of Rystsov digraphs parallel to restricted orbit digraphs of Section~\ref{sec:orbit digraph}. Namely, the  \emph{restricted Rystsov digraph} of $\mA=\langle\Zn,\{a,b\}\rangle$ is the spanning subgraph of $\Gamma_1(\mathrsfs{A})$ in which one retains only edges forced by words of length at most $n$. In order to transfer the arguments of Section~\ref{sec:orbit digraph} to perfectly reachable standardized DFAs, one needs to establish an analog of Proposition~\ref{prop:restricted scc}, that is, to show that what was connected in $\Gamma_1(\mathrsfs{A})$ remains so in the restricted Rystsov digraph. However, as the following example demonstrates, this is not true in general.

\begin{example}
\label{ex:restricted Rystsov}
Consider the standardized DFA $\mE_{12}=\langle\Z_{12},\{a,b\}\rangle$ shown in Figure~\ref{fig:e12}. Observe that $0\dt a=10=10\dt a$ in $\mE_{12}$ so that for this DFA, both parameters $d$ and $r$ are equal to 10 and the orbit $\orb(d)$ reduces to the singleton $\{10\}$. Therefore, the orbit digraph of  $\mE_{12}$ has two \scn\ components; they have as the vertex sets the subgroup $2\Z_{12}$ of all even residues modulo 12 and its coset $1\oplus2\Z_{12}$ consisting of all odd residues. In contrast, the Rystsov digraph $\ggam{1}{\mE_{12}}$ is strongly connected. This claim can be verified by either brute force successive checking through all words of defect 1 or invoking Propositions 2 and 3 of \cite{CaVo22} that characterize $\ggam{1}{\mE_{12}}$ as the Cayley digraph $\Cay(\Z_{12},D)$ where $D=\{d\dt v \mid v\in\{a,b^ra\}^*\}=\{10\dt v\mid v\in\{a,b^{10}a\}^*\}$. Going either way, one eventually finds the word $(ab^{10})^4a$ of length 45 that has defect~1 and forces the edge $0\to 1$ of $\ggam{1}{\mE_{12}}$. The word $(ab^{10})^4ab$ also has defect 1 and forces the edge $1\to 2$ in $\ggam{1}{\mE_{12}}$. The two edges $0\to 1$ and $1\to 2$ connect the \scn{} components of the orbit digraph of $\mE_{12}$, and thus, ensure strong connectivity of $\ggam{1}{\mE_{12}}$. By Proposition~\ref{prop:perfect} the automaton $\mE_{12}$ is perfectly reachable.

\begin{figure}[t]
\begin{center}
\begin{tikzpicture}
	\foreach \ev in {0,1,2,3,4,5,6,7,8,9,10,11}
	{
		\node[fill=white, circle, draw=blue, scale=1] at ($({30*(-\ev) + 90}:3cm)$) (\ev) {$\ev$};
		
	}	
	
	\draw
		(0) edge[-latex, dashed] (1)
		(1) edge[-latex, dashed] (2)
		(2) edge[-latex, dashed] (3)
		(3) edge[-latex, dashed] (4)
		(4) edge[-latex, dashed] (5)
		(5) edge[-latex, dashed] (6)
		(6) edge[-latex, dashed] (7)
		(7) edge[-latex, dashed] (8)
		(8) edge[-latex, dashed] (9)
		(9) edge[-latex, dashed] (10)
		(10) edge[-latex, dashed] (11)
		(11) edge[-latex, dashed] (0)
		
		(0) edge[-latex, bend right=60] (10)
		(10) edge[-latex, loop, out = 0, in = -60, distance = 1cm] (10)
        (11) edge[-latex, loop, out = -30, in = -90, distance = 1cm] (11)
        (3) edge[-latex, loop, out = 210, in = 150, distance = 1cm] (3)		
        (4) edge[-latex, loop, out = 180, in = 120, distance = 1cm] (4)		
        (5) edge[-latex, loop, out = 150, in = 90, distance = 1cm] (5)		
        (6) edge[-latex, loop, out = 120, in = 60, distance = 1cm] (6)		
        (7) edge[-latex, loop, out = 90, in = 30, distance = 1cm] (7)		
        (8) edge[-latex, loop, out = 60, in = 0, distance = 1cm] (8)		
        (9) edge[-latex, loop, out = 30, in = -30, distance = 1cm] (9)		
		(2) edge[-latex, bend right ] (1)
		(1) edge[-latex, bend right ] (2)
		
	;
\end{tikzpicture}
\end{center}
\caption{The DFA $\mE_{12}=\langle\Z_{12},\{a,b\}\rangle$. Solid and dashed edges show the action of $a$ and, resp., $b$.}\label{fig:e12}
\end{figure}
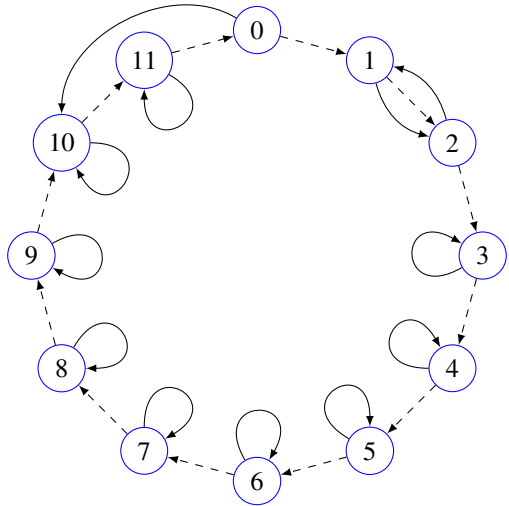

Computing all words of defect 1 and length at most 12, one gets the restricted Rystsov digraph of $\mE_{12}$ shown in Figure~\ref{fig:restrictedGamma e12}. This graph is not \scn.\qed
\end{example}

\begin{figure}[ht]
\begin{center}
\begin{tikzpicture}
	\foreach \ev in {0,2,4,6,8,10}
	{
		\node[fill=white, circle, draw=blue, scale=1] at ($({30*\ev +180}:2cm) + (-3,0)$) (\ev) {$\ev$};
	}
	
	\foreach \od in {1,3,5,7,9,11}
	{
		\node[fill=white, circle, draw=blue, scale=1] at ($({30*(\od -1) +180}:2cm) + (+3,0)$) (\od) {$\od$};
	}
	
	\draw
	(0) edge[-latex, left] node{$a$} (10)
	(2) edge[-latex, left] node{$ab^{2}$} (0)
	(4) edge[-latex, below] node{$ab^{4}$} (2)
	(6) edge[-latex, right] node{$ab^{6}$} (4)
	(8) edge[-latex, right] node{$ab^{8}$} (6)
	(10) edge[-latex, above] node{$ab^{10}$} (8)
	
	(1) edge[-latex,left] node{$ab$} (11)
	(3) edge[-latex,left] node{$ab^{3}$} (1)
	(5) edge[-latex, below] node{$ab^{5}$} (3)
	(7) edge[-latex, right] node{$ab^{7}$} (5)
	(9) edge[-latex, right] node{$ab^{9}$} (7)
	(11) edge[-latex, above] node{$ab^{11}$} (9)
	
	(0) edge[-latex, above, sloped] node{$ab^{10}a$} (8)
	;
\end{tikzpicture}
\end{center}
\caption{The restricted Rystsov digraph of the DFA $\mE_{12}$. Each edge is labeled by the shortest word of defect 1 forcing it.}\label{fig:restrictedGamma e12}
\end{figure}
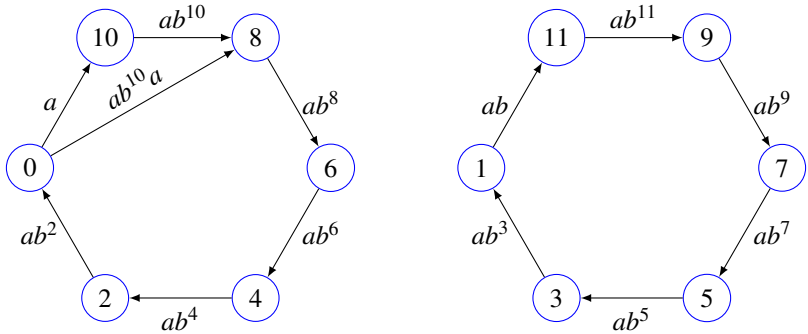

Thus, the method we used to prove Theorem~\ref{thm:good orbit} cannot be directly extended to show that Don's conjecture holds for perfectly reachable standardized DFAs. Of course, this does not disprove the conjecture. In particular, the automaton $\mE_{12}$ is not a counterexample to Don's conjecture. By Lemma~\ref{lem:n-extension}, to justify the latter claim, it suffices to show that every proper non-empty subset of $\Z_{12}$ is 12-expandable in $\mE_{12}$. Proposition~\ref{prop:noncosets} ensures this for all subsets except for $2\Z_{12}$ and $1\oplus2\Z_{12}$, and one easily verifies that the word $aba$ expands $1\oplus2\Z_{12}$ while the word $(ab)^2$ does the job for $2\Z_{12}$. (The words $aba$ and $(ab)^2$ have defect 2, and therefore, they do not show up when the restricted Rystsov digraph is constructed.)

Example~\ref{ex:restricted Rystsov} shows that expanding a subset with a short word can become possible if using words of defect greater than 1 is allowed. It is natural to ask whether this trick solves the issue for all perfectly reachable standardized DFAs. If so, then every proper non-empty subset in each $n$-state perfectly reachable standardized DFA would be $n$-expandable and Don's conjecture for such DFAs would follow by Lemma~\ref{lem:n-extension}. However, our next example exhibits a 21-state perfectly reachable standardized DFA with a subset that fails to be 21-expandable.

\begin{figure}[ht]
\begin{center}
\begin{tikzpicture}
	\node[fill=white, circle, draw=blue, draw, scale = 1] (0) at ({2*cos(18+10*360/10)},{2*sin(18+10*360/10)}) {$0$};
	\foreach \k [evaluate = \k as \e using int(2*\k)][evaluate = \k as \o using int(2*\k -1)] \k in {1,...,10}{
		
		\node[fill=white, circle, draw=blue, scale = 1] (\e) at ({4*cos(18+\k*360/10)},{4*sin(18+\k*360/10)}) {$\e$};
		\node[fill=white, circle, draw=blue, scale = 1] ({\o}) at ({3*cos(\k*360/10)},{3*sin(\k*360/10)}) {$\o$};
		
	}
	\foreach \k [evaluate =\k as \t using int(\k+1)] \k in{0,...,19}{
		\draw (\k) edge[-latex, dashed] (\t);
	}
		\draw
			(20) edge[-latex, dashed] (0)
			(0) edge[-latex, bend right = 10] (14)
			(7) edge[-latex, bend  right = 7] (18)
			(18) edge[-latex, bend  right= 7] (7)
			;
\end{tikzpicture}
\end{center}
\caption{The DFA $\mE_{21}=\langle\Z_{21},\{a,b\}\rangle$. Solid and dashed edges show the action of $a$ and, resp., $b$; if $a$ fixes a state, the corresponding loop is omitted to improve readability.}\label{fig:e21}
\end{figure}
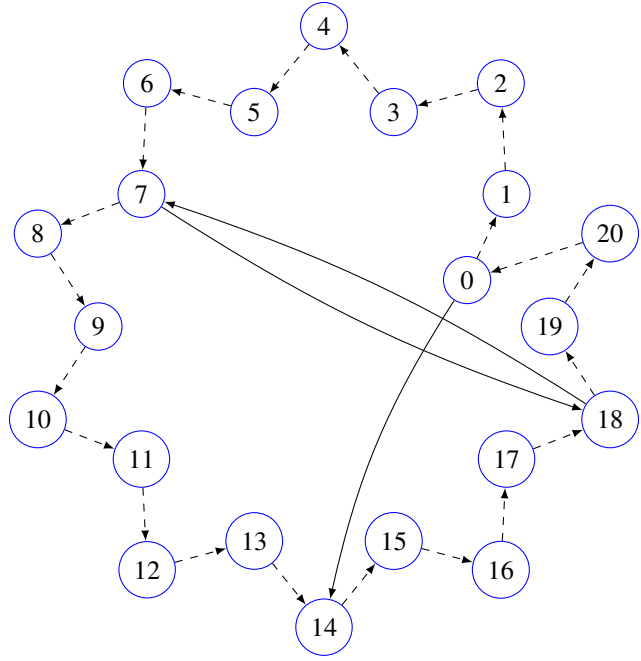

\begin{example}
\label{ex:nonexpandability}
Consider the standardized DFA $\mE_{21}=\langle\Z_{21},\{a,b\}\rangle$ shown in Figure~\ref{fig:e21} where all loops labeled $a$ are omitted to lighten the picture. Observe that $0\dt a=14=14\dt a$ in $\mE_{21}$. Thus, for $\mE_{21}$, both parameters $d$ and $r$ are equal to 14 and the orbit $\orb(d)$ reduces to the singleton $\{14\}$. The orbit digraph of $\mE_{21}$ has seven \scn\ components whose vertex sets are the subgroup $7\Z_{21}=\{0,7,14\}$ and its cosets. The word $ab^{14}a$ has defect 1, $\excl(ab^{14}a)=\{0\}$ and $\dupl(ab^{14}a)=\{18\}$. Hence, the Rystsov digraph $\ggam{1}{\mE_{21}}$ has the edge $0\to 18$. Multiplying $ab^{14}a$ on the right by $b,b^2,b^3,b^4,b^4,b^6$, we get words of defect 1 that force the edges
\[
1\to 19,\ 2\to 20,\ 3\to 0,\ 4\to 1,\ 5\to 2,\ 6\to 3
\]
in $\ggam{1}{\mE_{21}}$. These edges, together with $0\to 18$, cyclically connect all \scn{} components of the orbit digraph. Hence, the digraph $\ggam{1}{\mE_{21}}$ is strongly connected. By Proposition~\ref{prop:perfect} the automaton $\mE_{21}$ is perfectly reachable.

We have verified by brute force examination of all words in $\{a,b\}^*$ up to length 21 that none of them expand the subset $3\oplus7\Z_{21}=\{3,10,17\}$. The shortest word that expands $\{3,10,17\}$ is the word $ab^{14}ab^6$ of length 22.\qed
\end{example}

Although the automaton $\mE_{21}$ possesses a subset that is not 21-expandable, we have verified that it is not a counterexample to Don's conjecture. The only `bad' subset $\{3,10,17\}$ turns out to be the image of the word
\begin{equation}
\label{eq:long word}
(ab^{15}ab^3ab^4)^2 ab^4(ab^3(ab^4)^2)^2 ab^3 ab^4 ab^7 (ab^4)^2 ab^{14}ab^6.
\end{equation}
The length of the word \eqref{eq:long word} is 132, which is much less than the bound $21(21-3)=378$ claimed by Don's conjecture. The word \eqref{eq:long word} can be decomposed into a product of $18=21-3$ words of defect 1, but only the rightmost factor of defect 1 has length greater than 21 while all other factors are shorter, which by far compensates the excess of the last factor.

\section{Conclusion}
\label{sec:conclusion}

Using the concept of expandability, we have confirmed Don's conjecture for standardized DFAs subject to an arithmetical restriction to their orbits. Moreover, we have proved that in every standardized DFA $\langle\Zn,\{a,b\}\rangle$, almost all subsets are $n$-expandable; the only possible exceptions are the unions of cosets of the orbit subgroup of the DFA so that if the subgroup has index $k$ in $(\Zn,\oplus)$, then at least $2^n-2^k$ subsets of $\Zn$ are $n$-expandable. On the other hand, we found an example of a 21-state perfectly reachable standardized DFA with a subset that is not 21-expandable.

To our surprise, our results reveal that the situation around the expandability approach to Don's conjecture for perfectly reachable automata is in parallel with that around the extensibility approach to \v{C}ern\'{y}'s  conjecture for \sa. To discuss this similarity, recall a few definitions related to \sa.

A DFA $\mathrsfs{A}=\langle Q,\Sigma\rangle$ is \emph{synchronizing} if it has a reachable singleton, that is, $Q\dt w$ is a singleton for some word $w\in\Sigma^*$. Any such $w$ is said to be a \emph{reset word} for the DFA. We refer the reader to the chapter~\cite{KV21} of the `Handbook of Automata Theory' or to the second-named author's survey \cite{Volkov:2022} for an introduction to the rich theory of \sa. Here we only say that much research in this area groups around the famous \emph{\v{C}ern\'{y} conjecture} that every \san{} with $n$ states admits a reset word of length at most $(n-1)^2$. The conjecture remains open for almost 60 years though it has been confirmed for some species of DFAs. A method that has proved to be efficient for proving the \v{C}ern\'{y} conjecture for special classes of automata is based on the following notion. For a DFA $\mathrsfs{A}=\langle Q,\Sigma\rangle$, a subset $S\subseteq Q$ and a word $w\in\Sigma^*$, denote by $Sw^{-1}$ the full preimage of $S$ under the action of $w$, that is, $Sw^{-1}:=\{q\in Q\mid q\dt w\in S\}$. A word $w\in\Sigma^*$ is said to \emph{extend} a proper non-empty subset $S\subset Q$ if $|Sw^{-1}|>|S|$. Assuming that $|Q|=n$, we say that a proper non-empty subset $S\subset Q$ is $n$-\emph{extensible} if $S$ can be extended by a word of length at most $n$. It is well known (and easy to see) that each $n$-state DFA, all of whose proper non-empty subsets are $n$-extensible, is synchronizing and has a reset word of length at most $(n-1)^2$; that is, it fulfills the \v{C}ern\'{y} conjecture. The approach to the \v{C}ern\'{y} conjecture via $n$-extensibility traces back to Jean-\'Eric Pin's paper~\cite{Pin:1978}; the most striking applications of this approach are the proofs of the \v{C}ern\'{y} conjecture for circular \sa\ (Lois Dubuc~\cite{Dubuc:1998}) and \sa\ with Eulerian underlying digraphs (Jarkko Kari~\cite{Kari:2003}). On the other hand, Kari~\cite{Kari:2001} found an example of a 6-state \san\ with a subset that is not 6-extensible. This shows that, in general,  the \v{C}ern\'{y} conjecture cannot be proved via $n$-extensibility.

Clearly, every \cran\ is synchronizing (and this observation has served as one of the motivations for considering \cra). The converse is not true: it is easy to exhibit synchronizing, but not \cra{}, even in the class of standardized DFAs. As a concrete instance, consider the standardized DFA shown in Figure~\ref{fig:sync no cr}; it has $ab^4ab(ab^2)^2aba$ as a reset word, while the subset $\{0,1,3,4\}$ is not reachable.
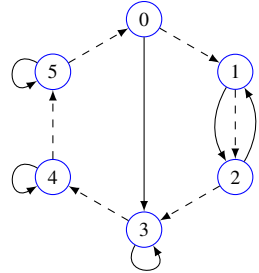
\begin{figure}[h]
\begin{center}
\begin{tikzpicture}
[scale=0.7]
	\foreach \i in {0,1,2,3,4,5}
	{
		\node[fill=white, circle, draw=blue, scale=0.75] at ($({90-60*\i }:2cm) + (-3,0)$) (\i) {$\i$};
	}
	
	\draw
		(0) edge[-latex, dashed]  (1)
		(1) edge[-latex, dashed]  (2)
		(2) edge[-latex, dashed]  (3)
		(3) edge[-latex, dashed]  (4)
		(4) edge[-latex, dashed]  (5)
		(5) edge[-latex, dashed]  (0)
		
		(0) edge[-latex]  (3)
		(3) edge[-latex, loop, out =240 , in =300, distance = 0.7cm ]  (3)
        (5) edge[-latex, loop, out =150 , in =210, distance = 0.7cm ]  (5)
        (4) edge[-latex, loop, out =150 , in =210, distance = 0.7cm ]  (4)
		(1) edge[-latex, bend right]  (2)
		(2) edge[-latex, bend right]  (1)
	;
\end{tikzpicture}
\end{center}
\caption{A standardized synchronizing DFA that is not completely reachable.  Solid and dashed edges show the action of $a$ and, resp., $b$.}\label{fig:sync no cr}
\end{figure}

Comparing the concepts of expanding and extending words, it is easy to see that if in a DFA $\mathrsfs{A}=\langle Q,\Sigma\rangle$, a word $w\in\Sigma^*$ expands a subset $S\subset Q$, then $w$ extends $S$ as well. The converse is not valid, even for standardized DFAs. As an instance, we can reuse the DFA in Figure~\ref{fig:sync no cr} where the word $ab^4ab(ab^2)^2aba$ extends the subset $\{0,1,3,4\}$ but does not expand it. Example~\ref{ex:nonexpandability} shows that even a standardized perfectly reachable DFA $\langle\Zn,\{a,b\}\rangle$ may have a subset that is not $n$-expandable, although all of its proper non-empty subsets are $n$-extensible by~\cite[Proposition 4.6]{Dubuc:1998}.

Thus, the notions of a perfectly reachable automaton and expandability are specializations of those of a \san{} and extensibility, respectively. Nevertheless, these specialized concepts behave similarly, to some extent, to how their more general counterparts do. One can, therefore, speculate that further advances in studying perfectly reachable automata and Don's conjecture may contribute to a better understanding of \sa{} and progress towards resolving the \v{C}ern\'{y} conjecture.

\medskip

\noindent\textbf{Added in revision.} Very recently, Yinfeng Zhu~\cite{Zhu:2024} has found an infinite series of binary completely reachable automata that violate Don's conjecture. Zhu's automata are not standardized; the validity of the conjecture for standardized completely reachable automata so far remains open. In~\cite{Zhu:2024}, it is also shown that every standardized DFA $\langle\Zn,\{a,b\}\rangle$ whose orbit subgroup coincides with the subgroup $2\Zn$ of the group $(\Zn,\oplus)$ fulfills Don's conjecture; this strengthens Theorem~\ref{thm:good orbit} of the present paper and is based on an enhancement of its techniques.

\medskip

\noindent\textbf{Acknowledgement.} We thank the anonymous reviewers for their careful reading of our paper and their many useful comments and suggestions.

\end{document}